\documentclass{article}
\usepackage[a4paper,margin=1.5in]{geometry}
\usepackage[latin1]{inputenc}
\usepackage{latexsym}
\usepackage{amsmath}
\usepackage{amsthm}
\usepackage{amssymb}
\usepackage{graphicx}
\usepackage{ifthen}
\usepackage{calc}
\usepackage{stmaryrd}
\usepackage{enumitem}
\usepackage{hyperref}
\usepackage{alphalph}
\usepackage{tikz-cd}

\usepackage{tikz}
\usetikzlibrary{calc}

\makeatletter
\renewcommand{\section}{\@startsection
  {section}%
  {1}%
  {0mm}%
  {-1\baselineskip}%
  {0.5\baselineskip}%
  {\normalfont\large\bfseries}%
}
\renewcommand{\subsection}{\@startsection
  {subsection}%
  {2}%
  {0mm}%
  {-1\baselineskip}%
  {0.5\baselineskip}%
  {\normalfont\large\itshape}%
}

\renewcommand{\subsubsection}{\@startsection
  {subsubsection}%
  {3}%
  {0mm}%
  {-1\baselineskip}%
  {0.5\baselineskip}%
  {\normalfont\itshape}%
}

\newsavebox{\tempbox}
\renewcommand{\@makecaption}[2]{
  \vspace{10pt}
  \sbox{\tempbox}{\textbf{#1.} #2}
  \ifthenelse{\lengthtest{\wd\tempbox > \linewidth}}{
    \textbf{#1.} #2\par
  }{
    \begin{center}
      \textbf{#1.} #2
    \end{center}
  }
}

\makeatother

\numberwithin{equation}{section}

\numberwithin{figure}{section}

\newtheoremstyle{mythm}
  {}
  {}
  {\itshape}
  {}
  {\bfseries}
  {.}
  {.5em}
  {\thmname{#1}~\thmnumber{#2}\ifthenelse{\equal{\thmnote{#3}}{}}{}{~(\thmnote{#3})}}

\newtheoremstyle{mydefn}
  {}
  {}
  {\upshape}
  {}
  {\bfseries}
  {.}
  {.5em}
  {\thmname{#1}~\thmnumber{#2}\ifthenelse{\equal{\thmnote{#3}}{}}{}{~(\thmnote{#3})}}

\newtheoremstyle{myremark}
  {}
  {}
  {\upshape}
  {}
  {\itshape}
  {.}
  {.5em}
  {\thmname{#1}~\thmnumber{#2}\ifthenelse{\equal{\thmnote{#3}}{}}{}{~(\thmnote{#3})}}

\theoremstyle{mythm}
\newtheorem{theorem}{Theorem}[section]
\newtheorem{lemma}[theorem]{Lemma}
\newtheorem{proposition}[theorem]{Proposition}
\newtheorem{corollary}[theorem]{Corollary}

\theoremstyle{mydefn}

\newtheorem{example}[theorem]{Example}

\theoremstyle{myremark}

\theoremstyle{mythm}

\newcommand{\uend}{\hfill$\lrcorner$}

\newcounter{claimcounter}
\newenvironment{claim}[1][]{
  \renewcommand{\proof}{\smallskip\par\noindent\textit{Proof. }}
  \medskip\par\noindent%
  \ifthenelse{\equal{#1}{}}{%
    \setcounter{claimcounter}{0}\refstepcounter{claimcounter}\textit{Claim~\arabic{claimcounter}.}
  }{%
    \ifthenelse{\equal{#1}{resume}}{%
      \refstepcounter{claimcounter}\textit{Claim~\arabic{claimcounter}.}
    }{%
      \textit{Claim~#1.}
    }
  }
}{
  \par\medskip
}

\newlist{caselist}{description}{10}
\setlist[caselist]{font=\itshape\mdseries}

\setenumerate[1]{label=(\arabic*)}
\newlist{eroman}{enumerate}{2}
\setlist[eroman,1]{label=(\roman*)}
\setlist[eroman,2]{label=(\alph*)}
\newlist{ealph}{enumerate}{1}
\setlist[ealph]{label=(\Alph*)}

\newcounter{nlistcounter}
\newenvironment{nlist}[1]{
  \renewcommand{\thenlistcounter}{\upshape(#1.\arabic{nlistcounter})}
  \begin{list}{\bfseries\thenlistcounter}{%
      \usecounter{nlistcounter}
      \setlength{\labelwidth}{1.5em}%
      \setlength{\leftmargin}{\labelwidth}%
      \addtolength{\leftmargin}{\labelsep}%
      \setlength{\listparindent}{0em}%
      \setlength{\topsep}{5pt}%
      \setlength{\itemsep}{5pt}%
      \setlength{\parsep}{0pt}%
    }
  }{
  \end{list}
}

\newcommand{\bigmid}{\;\big|\;}

\renewcommand{\mathbf}[1]{\textit{\bfseries #1}}

\renewcommand{\tilde}{\widetilde}

\renewcommand{\bar}{\overline}
\renewcommand{\vec}{\overrightarrow}

\renewcommand{\phi}{\varphi}
\renewcommand{\epsilon}{\varepsilon}

\newcommand{\NN}{{\mathbb N}}

\newcommand{\CC}{{\mathcal C}}

\newcommand{\CS}{{\mathcal S}}
\newcommand{\CT}{{\mathcal T}}

\newcounter{rbcounter}
\newcommand{\tw}{\operatorname{tw}}

\newcommand{\bw}{\operatorname{bw}}

\newcommand{\ord}{\operatorname{ord}}

\newcommand{\torso}[2]{#1\llbracket#2\rrbracket}

\begin{document}
\title{Tangles and Connectivity in Graphs}
\author{Martin Grohe\\\normalsize RWTH Aachen University\\\normalsize\texttt{grohe@informatik.rwth-aachen.de}}
\date{}

\maketitle

\begin{abstract}
  This paper is a short introduction to the theory of tangles, both in
  graphs and general connectivity systems. An emphasis is
  put on the correspondence between tangles of order $k$ and
  $k$-connected components. In particular, we prove that there is a
  one-to-one correspondence between the triconnected components of a
  graph and its tangles of order $3$.
 \end{abstract}

\section{Introduction}
Tangles, introduced by Robertson and Seymour in the tenth paper
\cite{gm10} of their graph minors series \cite{gmseries}, have come
to play an important part in structural graph theory. For example,
Robertson and Seymour's structure theorem for graphs with excluded
minors is phrased in terms of tangles in its general
form~\cite{gm16}. Tangles have
also played a role in algorithmic structural graph theory (for
example in
\cite{demhajkaw05,grokawree13,gromar12,groschwe15b,kawwol11}).

Tangles describe highly connected regions in a
graph. In a precise mathematical sense, they are ``dual'' to
decompositions (see Theorem~\ref{theo:duality}). Intuitively, a graph has a highly
connected region described by a tangle if and only if it does not
admit a decomposition along separators of low order. By
decomposition I always mean a decomposition in a treelike fashion;
formally, this is captured by the notions of tree decomposition or branch decomposition.

However, tangles describe regions of a graph in an indirect and
elusive way. This is why we use the unusual term ``region'' instead of
``subgraph'' or ``component''. The idea is that a tangle describes a region by
pointing to it. A bit more formally, a \emph{tangle of order $k$}
assigns a ``big side'' to every separation of order less than
$k$. The big side is where the  (imaginary) region described by the
tangle is supposed to be. Of course this assignment of ``big sides''
to the separations is subject to certain consistency and nontriviality
conditions, the ``tangle axioms''.

To understand why this way of describing a ``region'' is a good idea, let us review decompositions of
graphs into their $k$-connected components. It is well known that
every graph can be decomposed into its connected components and into
its biconnected components. The former are the (inclusionwise) maximal connected
subgraphs, and the latter the maximal 2-connected subgraphs. It is
also well-known that a graph can be decomposed into its triconnected
components, but the situation is more complicated here. Different from
what one might guess, the triconnected components are not maximal
3-connected subgraphs; in fact they are not even subgraphs, but just
topological subgraphs (see Section~\ref{sec:prel} for a definition of
topological subgraphs). Then what about 4-connected components?

 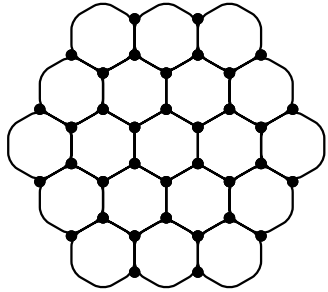
\begin{figure}
  \centering
\begin{tikzpicture}
  [
  current point is local=true,
  scale=0.8,
  line width=0.3mm,
  every node/.style={draw,circle,fill=black,inner sep=0mm,minimum
     size=1.2mm},
  every edge/.style={draw}
  ]

  \newcommand{\grd}{10.39230mm}
  \newcommand{\hex}{{
    +(90:6mm) node {} -- 
    +(150:6mm) node {} -- 
    +(210:6mm) node {} -- 
    +(270:6mm) node {} -- 
    +(330:6mm) node {} -- 
    +(30:6mm) node {} -- 
    cycle
  }}
  \newcommand{\hexref}{{
      [rounded corners]
    +(90:6mm) node {} --     
    +(30:6mm) node {} --     
    +(330:6mm) node {} --   
    +(270:6mm) node {} --   
    +(210:6mm) --   
    +(150:6mm) --   
    cycle
  }}
  \newcommand{\hexrde}{{
      [rounded corners]
    +(90:6mm) node {} --     
    +(30:6mm) node {} --     
    +(330:6mm) node {} --   
    +(270:6mm) --   
    +(210:6mm) --   
    +(150:6mm) node {} --   
    cycle
  }}
  \newcommand{\hexrcd}{{
      [rounded corners]
    +(90:6mm) node {} --     
    +(30:6mm) node {} --     
    +(330:6mm) --   
    +(270:6mm) --   
    +(210:6mm) node {} --   
    +(150:6mm) node {} --   
    cycle
  }}
  \newcommand{\hexrbc}{{
      [rounded corners]
    +(90:6mm) node {} --     
    +(30:6mm) --     
    +(330:6mm)--   
    +(270:6mm) node {} --   
    +(210:6mm) node {} --   
    +(150:6mm) node {} --   
    cycle
  }}
  \newcommand{\hexrab}{{
      [rounded corners]
    +(90:6mm) --     
    +(30:6mm) --     
    +(330:6mm) node {} --   
    +(270:6mm) node {} --   
    +(210:6mm) node {} --   
    +(150:6mm) node {} --   
    cycle
  }}
  \newcommand{\hexraf}{{
      [rounded corners]
    +(90:6mm) --     
    +(30:6mm) node {} --     
    +(330:6mm) node {} --   
    +(270:6mm) node {} --   
    +(210:6mm) node {} --   
    +(150:6mm) --   
    cycle
  }}
  \newcommand{\hexra}{{
      [rounded corners]
    +(90:6mm) --     
    +(30:6mm) node {} --     
    +(330:6mm) node {} --   
    +(270:6mm) node {} --   
    +(210:6mm)  node {} --   
    +(150:6mm)  node {} --   
    cycle
  }}
  \newcommand{\hexrb}{{
      [rounded corners]
    +(90:6mm) node {} --     
    +(30:6mm) --     
    +(330:6mm) node {} --   
    +(270:6mm) node {} --   
    +(210:6mm)  node {} --   
    +(150:6mm)  node {} --   
    cycle
  }}
  \newcommand{\hexrc}{{
      [rounded corners]
    +(90:6mm) node {} --     
    +(30:6mm) node {} --     
    +(330:6mm) --   
    +(270:6mm) node {} --   
    +(210:6mm)  node {} --   
    +(150:6mm)  node {} --   
    cycle
  }}
  \newcommand{\hexrd}{{
      [rounded corners]
    +(90:6mm) node {} --     
    +(30:6mm) node {} --     
    +(330:6mm) node {} --   
    +(270:6mm) --   
    +(210:6mm)  node {} --   
    +(150:6mm)  node {} --   
    cycle
  }}
  \newcommand{\hexre}{{
      [rounded corners]
    +(90:6mm)  node {} --     
    +(30:6mm) node {} --     
    +(330:6mm) node {} --   
    +(270:6mm) node {} --   
    +(210:6mm) --   
    +(150:6mm)  node {} --   
    cycle
  }}
  \newcommand{\hexrf}{{
      [rounded corners]
    +(90:6mm) node {} --     
    +(30:6mm) node {} --     
    +(330:6mm) node {} --   
    +(270:6mm) node {} --   
    +(210:6mm)  node {} --   
    +(150:6mm)  --   
    cycle
  }}

  \draw \hex 
  ++(240:\grd) \hex  
  ++(0:\grd) \hex
  ++(60:\grd) \hex
  ++(120:\grd) \hex
  ++(180:\grd) \hex
  ++(240:\grd) \hex  
  ++(240:\grd) \hexre 
  ++(300:\grd) \hexrde
  ++(0:\grd) \hexrd
  ++(0:\grd) \hexrcd
  ++(60:\grd) \hexrc
  ++(60:\grd) \hexrbc
  ++(120:\grd) \hexrb
  ++(120:\grd) \hexrab
  ++(180:\grd) \hexra
  ++(180:\grd) \hexraf
  ++(240:\grd) \hexrf  
  ++(240:\grd) \hexref  
  ;
 
\end{tikzpicture}

 \caption{A hexagonal grid}
  \label{fig:hex}
\end{figure}

It turns out that in general a graph does not have a reasonable decomposition
into 4-connected components (neither into $k$-connected components for
any $k\ge 5$), at least if these components are supposed to be
$4$-connected and some kind of subgraph. To understand the difficulty,
consider the hexagonal grid in Figure~\ref{fig:hex}. It is
3-connected, but not 4-connected. In fact, for any two nonadjacent
vertices there is a separator of order $3$ separating these two
vertices. Thus it is not clear what the 4-connected components of a
grid could possibly be (except, of course, just the single vertices,
but this would not lead to a meaningful decomposition). But maybe we
need to adjust our view on connectivity: a
hexagonal grid is fairly highly connected in a ``global sense''. All
its low-order separations are very unbalanced. In particular, all
separations of order $3$ have just a single vertex on one side and all
other vertices on the other side. This type of global connectivity is
what tangles are related to. For example, there is a unique tangle of
order $4$ in the hexagonal grid: the big side of a separation of order
$3$ is obviously the side that contains all but one vertex. The
``region'' this tangle describes is just the grid itself. This does not
sound particularly interesting, but the grid could be a subgraph of a
larger graph, and then the tangle would identify it as a highly
connected region within that graph. A key theorem about tangles is
that every graph admits a canonical tree decomposition into its
tangles of order $k$ \cite{cardiehar+13a,gm10}. This can be seen as a
generalisation of the decomposition of a graph into its 3-connected
components. A different, but related generalisation has been given in
\cite{cardiehun+14}.

The theory of tangles and decompositions generalises from graphs to an
abstract setting of \emph{connectivity systems}. This includes nonstandard
notions of connectivity on graphs, such as the 
``cut-rank'' function, which leads to the notion of ``rank
width'' \cite{oum05,oumsey06}, and connectivity functions on other
structures, for example matroids. Tangles give us an abstract notion of
``k-connected components'' for these connectivity
systems. The canonical decomposition theorem can be generalised from
graphs to this abstract setting~\cite{geegerwhi09,hun11}.

This paper is a short introduction to the basic theory of tangles,
both for graphs and for general connectivity systems. We put a
particular emphasis on the correspondence between tangles of order $k$ and
$k$-connected components of a graph for $k\le 3$, which gives some
evidence to the claim that for all $k$, tangles of order $k$ may be viewed as a
formalisation of the intuitive notion of ``$k$-connected
component''. 

The paper provides background material for my talk at LATA. The talk
itself will be concerned with more recent results \cite{gro16+} and, in particular,
computational aspects and applications of tangles~\cite{gromar15,groschwe15a,groschwe15b}.

\section{Preliminaries}
\label{sec:prel}

We use a standard terminology and notation (see \cite{die10} for background); let me just review a few
important notions. All graphs
considered in this paper are finite and simple. The vertex set
and edge set of a graph $G$ are denoted by $V(G)$ and $E(G)$,
respectively. The \emph{order} of $G$ is $|G|:=|V(G)|$. For a set
$W\subseteq V(G)$, we denote the \emph{induced subgraph} of $G$ with vertex
set $W$ by $G[W]$ and the induced subgraph with vertex set
$V(G)\setminus W$ by $G\setminus W$. 
The \emph{(open) neighbourhood} of a vertex $v$ in $G$ is denoted by
$N^G(v)$, or just $N(v)$ if $G$ is
clear from the context. For a set $W\subseteq V(G)$ we let
$
N(W):=\Big(\bigcup_{v\in W}N(v)\Big)\setminus W,
$
and for a subgraph $H\subseteq G$ we let $N(H):=N(V(H))$. 
The \emph{union} of two graphs $A,B$ is the graph $A\cup B$ with vertex set
$V(A)\cup V(B)$ and edge set $E(A)\cup E(B)$, and the \emph{intersection}
$A\cap B$ is defined similarly.

A \emph{separation} of $G$ is a pair $(A,B)$ of subgraphs of $G$ such
that $A\cup B=G$ and $E(A)\cap E(B)=\emptyset$. The \emph{order} of
the separation $(A,B)$ is $\ord(A,B):=|V(A)\cap V(B)|$. 
A separation $(A,B)$ is \emph{proper} if
$V(A)\setminus V(B)$ and $V(B)\setminus V(A)$ are both nonempty. 
A graph $G$ is \emph{$k$-connected} if $|G|>k$ and $G$ has no proper
$(k-1)$-separation. 

A \emph{subdivision} of $G$ is a graph obtained from $G$ by
subdividing some (or all) of the edges, that is, replacing them by
paths of length at least $2$. 
A graph $H$ is a \emph{topological subgraph} of $G$ if a
subdivision of $H$ is a subgraph of $G$. 

\section{Tangles in a Graph}
In this section we introduce tangles of graphs, give a few examples, and review
a few basic facts about tangles, all well-known and at least
implicitly from Robertson and Seymour's fundamental paper on tangles
\cite{gm10} (except Theorem~\ref{theo:reed}, which is due to Reed~\cite{ree97}).

Let $G$ be a graph.
A \emph{$G$-tangle} of order $k$ is a family $\CT$ of separations of
$G$ satisfying the following conditions.
\begin{nlist}{GT}
  \setcounter{nlistcounter}{-1}
\item\label{li:gt0}
  The order of all separations $(A,B)\in\CT$ is less than $k$.
\item\label{li:gt1}
  For all separations $(A,B)$ of $G$ of order less than $k$, either
  $(A,B)\in\CT$ or $(B,A)\in\CT$.
\item\label{li:gt2}
  If $(A_1,B_1), (A_2,B_2),(A_3,B_3)\in\CT$ then $A_1\cup A_2\cup
  A_3\neq G$.
\item\label{li:gt3}
  $V(A)\neq V(G)$ for all $(A,B)\in\CT$.
\end{nlist}
Observe that
\ref{li:gt1} and \ref{li:gt2} imply that for all separations $(A,B)$
of $G$ of order less than $k$, exactly one of the separations
$(A,B),(B,A)$ is in $\CT$.

We denote the order of a tangle $\CT$ by $\ord(\CT)$. 

\begin{example}\label{exa:tangle1}
  Let $G$ be a graph and $C\subseteq G$ a cycle. Let $\CT$ be the set
  of all separations $(A,B)$ of $G$ of order $1$ such that
  $C\subseteq B$. Then $\CT$ is a $G$-tangle of order $2$.
  
  To see this, note that $\CT$ trivially satisfies \ref{li:gt0}. It
  satisfies \ref{li:gt1}, because for every separation $(A,B)$ of $G$
  of order $1$, either $C\subseteq A$ or $C\subseteq B$. To see that
  $\CT$ satisfies \ref{li:gt3}, let $(A_i,B_i)\in\CT$ for
  $i=1,2,3$. Note that it may happen that $V(A_1)\cup V(A_2)\cup
  V(A_3)=V(G)$ (if $|C|=3$). However, no edge of $C$ can be in
  $E(A_i)$ for any $i$, because $C\subseteq B_i$ and $|A_i\cap B_i|\le
  1$. Hence $E(A_1)\cup A(A_2)\cup E(A_3)\neq E(G)$, which implies
  \ref{li:gt2}.
  Finally, $\CT$ satisfies \ref{li:gt3}, because $V(C)\setminus
  V(A)\neq\emptyset$ for all $(A,B)\in\CT$.
  \uend
\end{example}

\begin{example}
  Let $G$ be a graph and $X\subseteq V(G)$ a clique in $G$. Note that
  for all separations $(A,B)$ of $G$, either $X\subseteq V(A)$ or
  $X\subseteq V(B)$. For every $k\ge 1$, let $\CT_k$ be the
  set of all separations $(A,B)$ of $G$ of order less than $k$ such that
  $X\subseteq V(B)$.

  Then if $k<\frac{2}{3}|X|+1$, the set $\CT_k$ is a $G$-tangle of
  order $k$ . We omit the proof, which is similar to the proof in the
  previous example.

  Instead, we prove that $\CT_k$ is not necessarily a $G$-tangle if
  $k=\frac{2}{3}|X|+1$. To see this, let $G$ be a complete graph of
  order $3n$, $k:=2n+1$, and $X:=V(G)$. Suppose for contradiction that
  $\CT_k$ is a $G$-tangle of order $k$. Partition $X$ into three sets
  $X_1,X_2,X_3$ of size $n$. For $i\ne j$, let $A_{ij}:=G[X_i\cup X_j]$ and
  $B_{ij}:=G$. Then $(A_{ij},B_{ij})$ is a separation of $G$ of order
  $2n<k$. By \ref{li:gt1} and \ref{li:gt3}, we have
  $(A_{ij},B_{ij})\in\CT_k$. However, $A_{12}\cup A_{13}\cup
  A_{23}=G$, and this contradicts \ref{li:gt2}.
  \uend
\end{example}

\begin{figure}[t]
  \centering
  \begin{tikzpicture}
      [
      every node/.style={draw,circle,fill=black,inner sep=0mm,minimum
        size=2mm},
      every edge/.style={draw,thick},
      scale=0.7
      ]
      \foreach \x in {0,...,4}
         \foreach \y in {0,...,4}
            \fill (\x,\y) circle (3pt);
      \foreach \x in {0,...,4}
      {
         \draw[thick] (\x,0)--(\x,4);
         \draw[thick] (0,\x)--(4,\x);
       }
          
  \end{tikzpicture}
  \caption{A $(5\times 5)$-grid}\label{fig:grid}
\end{figure}
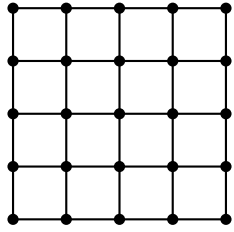

\begin{example}\label{exa:grid}
  Let $G$ be a graph and $H\subseteq G$ a $(k\times k)$-grid (see Figure~\ref{fig:grid}). Let
  $\CT$ be the set of all separations $(A,B)$ of $G$ of order at most
  $k-1$ such that $B$ contains some row of
  the grid. Then $\CT$ is a $G$-tangle of order $k$. (See \cite{gm10}
  for a proof.)
  \uend
\end{example}

The reader may wonder why in \ref{li:gt2} we take three separations, instead of two or four or seventeen. The following lemma
gives (some kind of) an explanation: we want our tangles to be closed
under intersection, in the weak form stated as assertion (3) of the
lemma; this is why taking just two separations in \ref{li:gt2} would
not be good enough. Three is just enough, and as we do not want to
be unnecessarily restrictive, we do not take more than three separations.

\begin{lemma}\label{lem:tangle-closure}
  Let $\CT$ be a $G$-tangle of order $k$.
  \begin{enumerate}
  \item If $(A,B)$ is a separation of $G$ with $|V(A)|<k$ then $(A,B)\in\CT$. 
  \item  If $(A,B)\in\CT$ and $(A',B')$ is a separation of $G$ of
    order $<k$ such that $B'\supseteq B$, then $(A',B')\in\CT$.
  \item  If $(A,B), (A',B')\in\CT$ and $\ord(A\cup A',B\cap B')<k$
    then $(A\cup A',B\cap B')\in\CT$.
  \end{enumerate}
\end{lemma}

\begin{proof}
  We leave the proofs of (1) and (2) to the reader. To prove (3), let
  $(A,B), (A',B')\in\CT$ and $\ord(A\cup A',B\cap B')<k$. By
  \ref{li:gt1}, either $(A\cup A',B\cap B')\in\CT$ or $(B\cup B',A\cap
  A')\in\CT$. As $A\cup A'\cup (B\cup B')=G$, by \ref{li:gt2} we cannot
  have $(B\cup B',A\cap
  A')\in\CT$.
  \qed
\end{proof}

\begin{corollary}\label{cor:tangle-closure}
  Let $\CT$ be a $G$-tangle of order $k$. Let
  $(A,B),(A',B')\in\CT$. Then $|B\cap B'|\ge k$.
\end{corollary}

The following lemma will allow us, among
other things, to give an alternative characterisation of tangles in
terms of so-called brambles.

\begin{lemma}
\label{lem:reed}
  Let $\CT$ be a $G$-tangle of order $k$. Then for every set
  $S\subseteq V(G)$ of cardinality $|S|<k$ there is a unique connected
  component $C(\CT,S)$ of $G\setminus S$ such that for all separations
  $(A,B)$ of $G$ with $V(A)\cap V(B)\subseteq S$ we have
  $(A,B)\in\CT\iff C(\CT,S)\subseteq B$.
\end{lemma}

\begin{proof}
    Let $C_1,\ldots,C_m$ be the set of all connected
  components of $G\setminus S$.  For every $I\subseteq [m]$, let
  $C_I:=\bigcup_{i\in I}C_i$. We define
  a separation $(A_I,B_I)$ of $G$ as follows. $B_I$ is the graph with
  vertex set $S\cup V(C_I)$ and all edges that have at
  least one endvertex in $V(C_I)$, and $A_I$ is the
  graph with vertex set $S\cup V(C_{[m]\setminus I})$ and
  edge set
   $E(G)\setminus E(B_I)$.
 Note that $V(A_I)\cap V(B_I)=S$ and thus $\ord(A_I,B_I)<k$. Thus for
 all $I$, either $(A_I,B_I)\in\CT$ or $(B_I,A_I)\in\CT$. It follows
 from Lemma~\ref{lem:tangle-closure}(1) and \ref{li:gt2} that
 $(B_I,A_I)\in\CT$ implies 
 $(A_{[m]\setminus I}, B_{[m]\setminus I})\in\CT$, because
 $(G[S],G)\in\CT$ and $B_I\cup B_{[m]\setminus I}\cup
 G[S]=G$. Furthermore, it follows from Lemma~\ref{lem:tangle-closure}(3) that $(A_I,B_I),(A_J,B_J)\in\CT$ implies
 $(A_{I\cap J},B_{I\cap J})\in\CT$.
 By \ref{li:gt3} we have $(A_{[m]},B_{[m]})\in\CT$ and
 $(A_\emptyset,B_\emptyset)\not\in\CT$.
 
 Let $I\subseteq[m]$ be of minimum cardinality such that
 $(A_I,B_I)\in\CT$. Since $(A_I,B_I),(A_J,B_J)\in\CT$ implies
 $(A_{I\cap J},B_{I\cap J})\in\CT$, the minimum set $I$ is unique. If
 $|I|=1$, then we let $C(\CT,S):=C_i$ for the unique element $i\in
 I$.
 Suppose for contradiction that $|I|>1$, and let $i\in I$. By the
 minimality of $|I|$ we have $(A_{\{i\}},B_{\{i\}})\not\in\CT$ and
 thus $(A_{[m]\setminus\{i\}},B_{[m]\setminus\{i\}})\in\CT$. This
 implies $(A_{I\setminus\{i\}},B_{I\setminus\{i\}})\in\CT$,
 contradicting the minimality of $|I|$.
\qed
\end{proof}

Let $G$ be a graph. We say that subgraphs $C_1,\ldots,C_m\subseteq G$
\emph{touch} if there is a vertex $v\in\bigcap_{i=1}^m V(C_i)$ or an edge
$e\in E(G)$ such that each $C_i$ contains at least one endvertex of
$e$. A family $\CC$ of subgraphs of $G$ \emph{touches pairwise} if all
$C_1,C_2\in\CC$ touch, and it \emph{touches triplewise} if all
$C_1,C_2,C_3\in\CC$ touch. A \emph{vertex cover} (or \emph{hitting set})
for $\CC$ is a set $S\subseteq V(G)$ such that $S\cap
V(C)\neq\emptyset$ for all $C\in\CC$.

\begin{theorem}[Reed \cite{ree97}]\label{theo:reed}
  A graph $G$ has a $G$-tangle of order $k$ if and only if there is a
  family $\CC$ of connected subgraphs of $G$ that touches triplewise
  and has no vertex cover of cardinality less than $k$.
\end{theorem}

In fact, Reed~\cite{ree97} defines a tangle of a graph $G$ to be a family
$\CC$  of connected subgraphs of $G$ that touches triplewise and its
order to be the cardinality of a minimum vertex cover. A
\emph{bramble} is a family
$\CC$  of connected subgraphs of $G$ that touches pairwise. In this
sense, a tangle is a special bramble. 

\begin{proof}[of Theorem~\ref{theo:reed}]
  For the forward direction, let $\CT$ be a $G$-tangle of order $k$.
  We let
  \[
  \CC:=\{C(\CT,S)\mid S\subseteq V(G)\text{ with }|S|<k\}.
  \]
  $\CC$ has no vertex cover of cardinality less than $k$, because if
  $S\subseteq V(G)$ with $|S|<k$ then $S\cap V(C(\CT,S))=\emptyset$. It
  remains to prove that $\CC$ touches triplewise. For $i=1,2,3$, let
  $C_i\in\CC$ and $S_i\subseteq V(G)$ with $|S_i|<k$ such that
  $C_i=C(\CT,S_i)$. Let $B_i$ be the graph with vertex set $V(C_i)\cup
  S$ and all edges of $G$ that have at least one vertex in $V(C_i)$,
  and let $A_i$ be the graph with vertex set $V(G)\setminus V(C_i)$
  and the remaining edges of $G$. Since
  $C(\CT,S_i)=C_i\subseteq B_i$, we have $(A_i,B_i)\in\CT$. Hence $A_1\cup A_2\cup
  A_3\neq G$ by \ref{li:gt2}, and this implies that $C_1,C_2,C_3$
  touch.

  \medskip
  For the backward direction, let $\CC$ be a family of connected
  subgraphs of $G$ that touches triplewise
  and has no vertex cover of cardinality less than $k$.
  We let $\CT$ be the set of all separations $(A,B)$ of $G$ of order less
  than $k$ such that $C\subseteq B\setminus V(A)$ for some
  $C\in\CC$. It is easy to verify that $\CT$ is a $G$-tangle of order
  $k$.
  \qed
\end{proof}

Let
$\CT,\CT'$ be $\kappa$-tangles. If $\CT'\subseteq\CT$, we say that
$\CT$ is an \emph{extension} of $\CT'$. 
The \emph{truncation} 
of $\CT$ to order $k\le\ord(\CT)$ is the set
$
\{(A,B)\in\CT\mid\ord(A,B)<k\},
$
which is obviously a tangle of order $k$. Observe that if $\CT$ is
an extension of $\CT'$, then $\ord(\CT')\le\ord(\CT)$, and $\CT'$ is
the truncation of $\CT$ to order $\ord(\CT')$. 

\section{Tangles and Components}

In this section, we will show that there is a one-to-one correspondence between the
tangles of order at most $3$ and the connected, biconnected, and
triconnected components of a graph. Robertson and Seymour~\cite{gm10}
established a one-to-one correspondence between tangles of order $2$ and
biconnected component. Here, we extend the picture to tangles of order
$3$.\footnote{My guess is that the result for tangles of order $3$ is known
  to other researchers in the field, but I am not aware of it being
  published anywhere.}

\subsection{Biconnected and Triconnected Components}

Let $G$ be a graph. Following \cite{cardiehun+14}, we call a set $X\subseteq V(G)$
\emph{$k$-inseparable} in $G$ if $|X|>k$ and there is no separation $(A,B)$ of $G$ of
order at most $k$ such that $X\setminus V(B)\neq\emptyset$ and
$X\setminus V(A)\neq\emptyset$. A \emph{$k$-block} of $G$ is an inclusionwise
maximal $k$-inseparable subset of $V(G)$. We call a $k$-inseparable
set of cardinality greater than $k+1$ a \emph{proper} $k$-inseparable
set and, if it is a $k$-block, a \emph{proper} $k$-block. (Recall that
a $(k+1)$-connected graph has order greater than $k+1$ by definition.)
We observe that every vertex $x$
in a proper $k$-inseparable set $X$ has degree at least $(k+1)$, because it
has $(k+1)$ internally disjoint paths to $X\setminus\{x\}$. 

A \emph{biconnected component} of $G$ is a subgraph induced by a
1-block, which is usually just called a \emph{block}.\footnote{There
  is a slight discrepancy to standard terminology here: a set
  consisting of a single isolated
  vertex is usually also called a block, but it is not a
  $1$-block, because its size is not greater than $1$.} It is easy to
see that a biconnected component $B$ either consists of a single edge
that is a bridge of $G$, or it is 2-connected. 
In the latter case, we call $B$ a \emph{proper biconnected
  component}.

The definition of triconnected components is more
complicated, because the subgraph induced by a 2-block is not
necessarily $3$-connected (even if it is a proper 2-block). 

\begin{example}
  Let $G$ be a graph obtained from the complete graph $K_4$ by
  subdividing each edge once. Then the vertices of the original $K_4$,
  which are precisely the vertices of degree $3$ in $G$, form a
  proper 2-block, but the subgraph they induce has no edges and thus is
  certainly not 3-connected.
\end{example}

It can be shown, however, that every proper
2-block of $G$ is the vertex set of a 3-connected topological
subgraph.
For a subset $X\subseteq V(G)$, we define the \emph{torso} of $X$ in
$G$ to be the graph $\torso{G}{X}$ obtained from the induced subgraph
$G[X]$ by adding an edge $vw$ for all distinct $v,w\in X$ such that there is a
connected component $C$ of $G\setminus X$ with $v,w\in N(C)$. We call
the edges in $E(\torso GX)\setminus E(G)$ the \emph{virtual edges} of
$\torso GX$. It is not hard to show that if $X$ is a 2-block of $G$
then for every connected component $C$ of $G\setminus X$ it holds that
$N(C)\le 2$; otherwise $X$
would not be an \emph{inclusionwise maximal} 2-inseparable set. This
implies that $\torso GX$ is a topological subgraph of $G$: if, for
some connected component $C$ of $G\setminus X$, 
$N(C)=\{v,w\}$ and hence $vw$ is a virtual edge of the torso,  then
there is a path from $v$ to $w$ in $C$, which may be viewed as a
subdivision of the edge $vw$ of $\torso GX$.
We call the torsos $\torso{G}{X}$ for the 2-blocks $X$ the
\emph{triconnected components} of $G$. We call a
triconnected component \emph{proper} if its order is at least $4$.

It is a well known fact, going back to MacLane~\cite{mac37} and Tutte~\cite{tut84}, that all graphs admit tree decompositions into
their biconnected and triconnected components.  Hopcroft and
Tarjan~\cite{tar72,hoptar73a} proved that the decompositions can be computed
in linear time.

\subsection{From Components to Tangles}

\begin{lemma}\label{lem:t1}
  Let $G$ be a graph and $X\subseteq V(G)$ a $(k-1)$-inseparable set of
  order $|X|>\frac{3}{2}\cdot (k-1)$. Then
  \[
  \CT^{(k)}(X):=\{(A,B)\bigmid(A,B)\text{ separation of $G$ of order
    $<k$ with }X\subseteq V(B)\big\}
  \]
  is a $G$-tangle of order $k$. 
\end{lemma}

\begin{proof}
$\CT^{(k)}(X)$ trivially satisfies
  \ref{li:gt0}. It satisfies \ref{li:gt1}, because the
  $(k-1)$-inseparability of $X$ implies that for every separation $(A,B)$ of
  $G$ of order $<k$ either $X\subseteq V(A)$ or $X\subseteq
  V(B)$.

  To see that $\CT^{(k)}(X)$ satisfies \ref{li:gt2}, let
  $(A_i,B_i)\in\CT^{k}(X)$ for $i=1,2,3$. Then $|V(A_i)\cap X|\le
  k-1$, because $V(A_i)\cap X\subseteq V(A_i)\cap V(B_i)$. 
  As $|X|>\frac{3}{2}\cdot (k-1)$, there is a vertex $x\in X$
  such that $x$ is contained in at most one of the sets $V(A_i)$. Say,
  $x\not\in V(A_2)\cup V(A_3)$. If $x\not\in V(A_1)$, then
  $V(A_1)\cup V(A_2)\cup V(A_3)\neq V(G)$. So let us assume that
  $x\in V(A_1)$.  

  Let $y_1,\ldots,y_{k-1}\in X\setminus\{x\}$. As $X$ is
  $(k-1)$-inseparable, for all $i$ there is a path $P_i$ from $x$ to $y_i$
  such that $V(P_i)\cap V(P_j)=\{x\}$ for $i\neq j$. Let
  $w_i$ be the last vertex of $P_i$ (in the direction from $x$ to
  $y_i$) that is in $V(A_1)$. We claim that $w_i\in V(B_1)$. This is the case if $w_i=y_i\in  X\subseteq
  V(B_1)$. If $w_i\neq y_i$, let $z_i$ be the successor of $w_i$ on
  $P_i$. Then $z_i\in V(B_1)\setminus V(A_1)$, and as $w_iz_i\in
  E(G)$, it follows that $w_i\in V(B_1)$ as well.

  Thus $\{x,w_1,\ldots,w_{k-1}\}\subseteq V(A_1)\cap V(B_1)$, and as
  $|V(A_1)\cap V(B_1)|\le k-1$, it follows that $w_i=x$ for some
  $i$. Consider the edge $e=xz_i$. We have $e\not\in E(A_1)$ because
  $z_i\not\in V(A_1)$ and $e\not\in E(A_2)\cup E(A_3)$ because
  $x\not\in V(A_2)\cup V(A_3)$. Hence $E(A_1)\cup E(A_2)\cup
  E(A_3)\neq E(G)$, and this completes the proof of \ref{li:gt2}.

  Finally, $\CT^{(k)}(X)$ satisfies
  \ref{li:gt3}, because for every $(A,B)\in\CT$ we have $|V(A)\cap X|\le k-1<|X|$.
\qed
\end{proof}

\begin{corollary}\label{cor:block->tangle}
  Let $G$ be a graph and $X\subseteq V(G)$.
  \begin{enumerate}
  \item If $X$ is the vertex set of a connected component of $G$
    (that is, a $0$-block), then $\CT^1(X)$ is a $G$-tangle of
    order $1$.
  \item If $X$ is the vertex set of a biconnected component of $G$
    (that is, a $1$-block), then $\CT^2(X)$ is a $G$-tangle of 
    order $2$. 
  \item If $X$ is the vertex set of a proper triconnected component of $G$
    (that is, a $2$-block of cardinality at least $4$), then $\CT^3(X)$ is a $G$-tangle of 
    order $3$. 
  \end{enumerate}
\end{corollary}

Let us close this section by observing that the restriction to
\emph{proper} triconnected components in assertion (3) of the
corollary is necessary.

\begin{lemma}\label{lem:exc}
  Let $G$ be a graph and $X\subseteq V(G)$ be a $2$-block of
  cardinality $3$. Then $\CT^3(X)$ is not a tangle.
\end{lemma}

\begin{proof}
  Let $\CT:=\CT^3(X)$.
  Suppose that $X=\{x_1,x_2,x_3\}$. For $i\neq j$, let
    $S_{ij}:=\{x_i,x_j\}$, and let $Y_{ij}$ be the union of the vertex
    sets of all connected components $C$ of $G\setminus X$ with
    $N(C)\subseteq S_{ij}$, and let
    $Z_{ij}:=V(G)\setminus (Y_{ij}\cup S_{ij})$. Let
    $A_{ij}:=G[Y_{ij}\cup S_{ij}]$, and let $B_{ij}$ be the graph with
    vertex set $S_{ij}\cup Z_{ij}$ and edge set
    $E(G)\setminus E(A_{ij})$.  Then $(A_{ij},B_{ij})\in\CT$, because
    $X\subseteq V(B_{ij})$.  As 
    $X$ is a 2-block, for every connected component $C$ of
    $G\setminus X$ it holds that $|N(C)|\le2$, and hence
    $C\subseteq A_{ij}$ for some $i,j$. It is not hard to see that
    this implies $A_{12}\cup A_{13}\cup A_{23}=G$. Thus $\CT$ violates
    \ref{li:gt2}.
    \qed
\end{proof}
  
\subsection{From Tangles to Components}
For a $G$-tangle $\CT$, we let 
\[
X_{\CT}:=\bigcap_{(A,B)\in\CT} V(B).
\]
In general, $X_{\CT}$ may be empty; an example is the tangle of order
$k$ associated with a $(k\times k)$-grid for $k\ge 5$ (see
Example~\ref{exa:grid}). However, it turns out that for tangles of
order $k\le 3$, the set $X_{\CT}$ is a $(k-1)$-block. This will be the main
result of this section.

\begin{lemma}\label{lem:X}
  Let $\CT$ be a $G$-tangle of order $k$. If $|X_{\CT}|\ge k$, then
  $X_{\CT}$ is a $(k-1)$-block of $G$ and $\CT=\CT^k(X_{\CT})$.
\end{lemma}

\begin{proof}
  Suppose that $|X_{\CT}|\ge k$.  If $(A,B)$ is a
  separation of $G$ of order less than $k$ then either $(A,B)\in\CT$ or
  $(B,A)\in\CT$, which implies $X_{\CT}\subseteq V(B)$ or
  $X_{\CT}\subseteq V(A)$. Thus $X_{\CT}$ is $(k-1)$-inseparable.
  If $X\supset X_{\CT}$,
  say, with $x\in X\setminus X_{\CT}$, then there is some separation
  $(A,B)\in\CT$ with $x\in V(A)\setminus V(B)$ and $X_{\CT}\subseteq
  V(B)$, and this implies that $X$ is not $(k-1)$-inseparable. Hence
  $X_{\CT}$ is a $k$-block.

  We have $\CT=\CT^k(X_\CT)$, because $X_{\CT}\subseteq V(B)$ for all
  $(A,B)\in\CT$, and for a separation $(A,B)$ of order at most $k-1$
  we cannot have $X_{\CT}\subseteq V(A)\cap V(B)$.
  \qed
\end{proof}

Let $\CT$ be a $G$-tangle.  A separation $(A,B)\in\CT$ is
\emph{minimal} in $\CT$ if there is no $(A',B')\in\CT$ such that
$B'\subset B$. Clearly, $X_{\CT}$ is the intersection of all sets
$V(B)$ for minimal $(A,B)\in\CT$. Hence if we want to understand
$X_{\CT}$, we can restrict our attention to the minimal separations in $\CT$.
Let $(A,B)\in\CT$ be minimal and $S:=V(A)\cap V(B)$. It
follows from Lemma~\ref{lem:reed} that $B\setminus S= C:=C(\CT,S)$,
and it follows from the minimality that $S=N\big(C)$ and that $E(B)$
consists of all edges with one endvertex in $V(C)$. Hence $B$ is
connected.

\begin{theorem}[Robertson and Seymour~\cite{gm10}]\label{theo:2tangle->block}
   Let $G$ be a graph.
  \begin{enumerate}
  \item For every $G$-tangle $\CT$ of order $1$, the set $X_{\CT}$ is
    a vertex set of a connected component of $G$, and we have $\CT=\CT^1(X_{\CT})$.
  \item For every $G$-tangle $\CT$ of order $2$, the set $X_{\CT}$ is
    the vertex set of a biconnected component of $G$, and we have $\CT=\CT^2(X_{\CT})$.
  \end{enumerate}
\end{theorem}

\begin{proof}
  To prove (1), let $\CT$ be a $G$-tangle of order $1$. Let
  $C=C(\CT,\emptyset)$. Then $(G\setminus V(C),C)$ is the unique
  minimal separation in $\CT$, and thus we have $X_{\CT}=V(C)$.

  To prove (2), let $\CT$ be a $G$-tangle of order $2$. By
  Lemma~\ref{lem:X}, it suffices to prove that $|X_{\CT}|\ge 2$. Let $\CT'$
  be the truncation of $\CT$ to order $1$. Then $W:=X_{\CT'}$ is the
  vertex set of a connected component $C$ of $G$, and we have
  $X_{\CT}\subseteq W$. Moreover, for every minimal $(A,B)\in\CT$ we
  have $B\subseteq C$, because $B$ is connected and $B\cap
  C\neq\emptyset$ by \ref{li:gt2}. 

  \begin{figure}
    \centering
    \begin{tikzpicture}[scale=0.8]

      \begin{scope}
        \fill[black!10] (-1.75,-1.75) rectangle (1.75,1.75);
        \fill[black!20] (-1.75,-0.25) rectangle (1.75,0.25);
        \fill[black!20] (-0.25,-1.75) rectangle (0.25,1.75);

      \draw[thick] (-1.75,-1.75) rectangle (1.75,1.75)
                   (-0.25,-1.75) rectangle (0.25,1.75)
                   (-1.75,-0.25) rectangle (1.75,0.25);
      
      \path (-1.76,1) node[anchor=east] {$B_1$}             
            (-1.76,0) node[anchor=east]  {$A_1\cap B_1$}        
            (-1.76,-1) node[anchor=east]  {$A_1$}  
            (-1,2) node {$B_2$}
            (0,2.5) node[rotate=270] {$A_2\cap B_2$}
            (1,2) node {$A_2$}
      ;

      \fill (-1,0) circle (2pt) node[right] {$s_1$};
      \fill (0,1) circle (2pt) node[below] {$s_2$};;
      \end{scope}

    \end{tikzpicture}
    \caption{Proof of Theorem~\ref{theo:2tangle->block}}
    \label{fig:cross2}
  \end{figure}
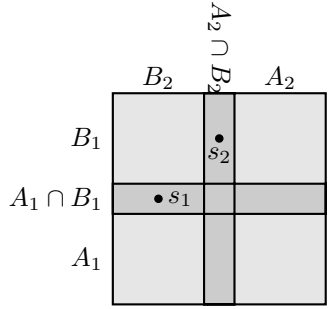

  \begin{claim} 
    Let $(A_1,B_1),(A_2,B_2)\in\CT$ be distinct and minimal in $\CT$. Then
    $A_1\cap C\subseteq B_2$ and $A_2\cap C\subseteq B_1$.

    \proof
    We have $\ord(A_1\cup A_2,B_1\cap B_2)\ge 2$, because otherwise
    $(A_1\cup A_2,B_1\cap B_2)\in\CT$ by
    Lemma~\ref{lem:tangle-closure}(3), which contradicts the
    minimality of the separations $(A_i,B_i)$. Suppose that $V(A_i)\cap
    V(B_i)=\{s_i\}$. As 
    \[
    V(A_1\cup A_2)\cap( V(B_1\cap
    B_2)\subseteq V(A_1\cap B_1)\cup V(A_2\cap
    B_2)=\{s_1,s_2\},
    \]
    we must have $s_1\neq s_2$ and $V(A_1\cup A_2)\cap V(B_1\cap
    B_2)=\{s_1,s_2\}$ (see Figure~\ref{fig:cross2}). This implies $V(A_1\cap A_2)\cap
    V(B_1\cup B_2)=\emptyset$. Then $(A_1\cap A_2,B_1\cup B_2)$ is a
    separation of $G$ of order $0$, and as $C$ is connected and
    $(B_1\cup B_2)\cap C\neq\emptyset$, we have $A_1\cap A_2\cap
    C=\emptyset$. The assertion of the claim follows.
    \uend
  \end{claim}

  Let $(A_1,B_1),\ldots,(A_m,B_m)$ be an enumeration of all minimal
  separations in $\CT$ of order $1$. Even if $C$ is $1$-inseparable,
  there is such a separation:
  $\big(G\setminus (V(C)\setminus\{v\}),C\big)$ for an arbitrary
  $v\in V(C)$.  Thus $m\ge 1$. If $m=1$, then $X_{\CT}=V(B_1)$ and
  thus $|X_{\CT}|\ge 2$ by Lemma~\ref{lem:tangle-closure}(1).

  If $m\ge 2$, let $A_i\cap B_i=\{s_i\}$. We can
  assume the $s_i$ to be mutually distinct, because if $s_i=s_j$ then
  $B_i=B_j$. It follows from Claim~1 that
  $s_1,\ldots,s_m\in \bigcap_iV(B_i)=X_{\CT}$. This implies
  $|X_{\CT}|\ge 2$.  \qed
\end{proof}

To extend Theorem~\ref{theo:2tangle->block} to tangles of order $3$,
we first prove a lemma, which essentially says that we can restrict our
attention to 2-connected graphs. Let $G$ be graph and
$X\subseteq V(G)$. For every $A\subseteq G$, let
$A\cap X:=A\big[V(A)\cap X]$. Note that if $(A,B)$ is a separation of
$G$, then $(A\cap X,B\cap X)$ is a separation of $G[X]$ with
$\ord(A\cap X,B\cap X)\le\ord(A,B)$.

\begin{lemma}
  Let $\CT$ be a $G$-tangle of order $3$. Let $\CT'$ be the truncation
  of $\CT$ to order $2$, and let $W:=X_{\CT'}$. Let $\CT[W]$ be the
  set of all separations $(A\cap W ,B\cap W)$ of $G[W]$ where
  $(A,B)\in\CT$. Then $\CT[W]$ is a $G[W]$-tangle of order
  $3$. Furthermore, $X_{\CT}=X_{\CT[W]}$.
\end{lemma}

\begin{proof}
  By Theorem~\ref{theo:2tangle->block}, $G[W]$ is a
  biconnected component of $G$. This implies that
  $|W|\ge 2$ and $|N(C)|\le1$ for every connected component $C$ of
  $G\setminus W$. For every $w\in W$, we let $Y_w$ be union of the vertex
  sets of all connected components $C$ of $G\setminus W$ with
  $N(C)\subseteq \{w\}$. Then $V(G)=W\cup\bigcup_{w\in W}Y_w$. Let $Z_w:=V(G)\setminus
  (Y_w\cup\{w\})$.
  Let $A_w:=G[Y_w\cup\{w\}]$ and $B_w:=G[Z_w\cup\{w\}]$. Then
  $W\subseteq V(B_w)$ and thus
  $(A_w,B_w)\in\CT^2(W)=\CT'\subseteq\CT$.

  \begin{claim}
    Let $(A,B)\in\CT$. Then $W\setminus
    V(A)\neq\emptyset$.


    \proof
    Suppose for contradiction that $W\subseteq V(A)$. Let $S:=V(A)\cap
    V(B)$ and suppose that
    $S=\{s_1,s_2\}$. Let $w_i\in W$ such that $s_i\in
    Y_{w_i}\cup\{w_i\}$. Then $A\cup A_{w_1}\cup A_{w_2}=G$, which
    contradicts \ref{li:gt2}. This proves that $W\setminus
    V(A)\neq\emptyset$. 
    \uend
  \end{claim}

  It is now straightforward to prove that $\CT[W]$ satisfies the
  tangle axioms \ref{li:gt0}, \ref{li:gt1}, and \ref{li:gt3}. To prove
  \ref{li:gt2}, let $(A_i,B_i)\in\CT$ for $i=1,2,3$. We need to prove
  that $(A_1\cap W)\cup (A_2\cap W)\cup (A_3\cap W)\neq G[W]$. Without
  loss of generality we may assume that $(A_i,B_i)$ is minimal in
  $\CT$. Then $C_i:=B_i\setminus V(A_i)$ is connected. By Claim~1,
  $V(C_i)\cap W\neq\emptyset$. This implies that if $V(C_i)\cap
  Y_w\neq\emptyset$ for some $w\in W$, then $w\in V(C_i)$. 

  As $\CT$ satisfies \ref{li:gt2}, $A_1\cup A_2\cup A_3\neq G$, and
  thus there either is a vertex in $V(C_1)\cap V(C_2)\cap V(C_3)$ or
  an edge with an endvertex in every $V(C_i)$. Suppose first that
  $v\in V(C_1)\cap V(C_2)\cap V(C_3)$. If $v\in W$ then
  \[
  V\big((A_1\cap W)\cup (A_2\cap W)\cup (A_3\cap W)\big)\neq
  W=V\big(G[W]\big).
  \]
  Otherwise, $v\in Y_w$ for some $w\in W$, and we
  have $w\in V(C_1)\cap V(C_2)\cap V(C_3)$. Similarly, if $e=vv'$ has an endvertex in
  every $V(C_i)$, then we distinguish between the case that $v,v'\in W$,
  which implies $E\big((A_1\cap W)\cup (A_2\cap W)\cup (A_3\cap W)\big)\neq
  E\big(G[W]\big)$, and the case that $e\in E(A_w)$ for some $w\in W$,
  which implies $w\in V(C_1)\cap V(C_2)\cap V(C_3)$ and thus $V\big((A_1\cap W)\cup (A_2\cap W)\cup (A_3\cap W)\big)\neq
  W=V\big(G[W]\big)$. This proves
  \ref{li:gt2} and hence that $\CT[W]$ is a tangle.

The second assertion $X_{\CT}=X_{\CT[W]}$ follows from the fact that
$X_{\CT}\subseteq X_{\CT'}=W$.
\qed
\end{proof}

\begin{theorem}\label{theo:3tangle->block}
  Let $G$ be a graph. For every $G$-tangle $\CT$ of order $3$, the set
  $X_{\CT}$ is a vertex set of a proper triconnected component of $G$.
\end{theorem}

\begin{proof}
  Let $\CT$ be a $G$-tangle of order $3$. 
  It suffices to prove that $|X_{\CT}|\ge 3$. Then by
  Lemma~\ref{lem:X}, $X_{\CT}$ is a 3-block and $\CT=\CT^3(X_{\CT})$,
  and by Lemma~\ref{lem:exc}, $X_{\CT}$ is proper 3-block, that is,
  the vertex set of a proper triconnected component.

  By the previous lemma, we may assume without loss of generality that
  $G$ is 2-connected.  The rest of the proof follows the lines of the proof of
  Theorem~\ref{theo:2tangle->block}. The core of the proof is again an
  ``uncrossing argument'' (this time a more complicated one) in
  Claim~1.

  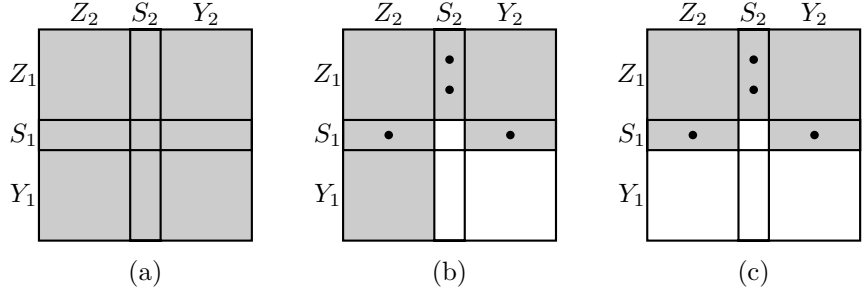
\begin{figure}
    \centering
    \begin{tikzpicture}[scale=0.8]

      \begin{scope}
        \fill[black!20] (-1.75,-1.75) rectangle (1.75,1.75);
      \draw[thick] (-1.75,-1.75) rectangle (1.75,1.75)
                   (-0.25,-1.75) rectangle (0.25,1.75) 
                   (-1.75,-0.25) rectangle (1.75,0.25);
      
      \path (-2,1) node {$Z_1$}             
            (-2,0) node {$S_1$}        
            (-2,-1) node {$Y_1$}  
            (-1,2) node {$Z_2$}
            (0,2) node {$S_2$}
            (1,2) node {$Y_2$}
      ;
      
      \path (0,-2.3) node {(a)};
      \end{scope}

     \begin{scope}[xshift=5cm]
       
       \fill[black!20] (-1.75,-1.75) rectangle (-0.25,1.75)
       (-0.25,0.25) rectangle (0.25,1.75)
       (0.25,-0.25) rectangle (1.75,1.75)
       ;

      \draw[thick] (-1.75,-1.75) rectangle (1.75,1.75)
                   (-0.25,-1.75) rectangle (0.25,1.75) 
                   (-1.75,-0.25) rectangle (1.75,0.25);
      
      \path (-2,1) node {$Z_1$}             
            (-2,0) node {$S_1$}        
            (-2,-1) node {$Y_1$}  
            (-1,2) node {$Z_2$}
            (0,2) node {$S_2$}
            (1,2) node {$Y_2$}
      ;

      \fill (-1,0) circle (2pt)  (1,0) circle (2pt)  (0,0.75) circle
      (2pt) (0,1.25) circle (2pt);

      \path (0,-2.3) node {(b)};
       \end{scope}

     \begin{scope}[xshift=10cm]
       
       \fill[black!20] (-1.75,-0.25) rectangle (-0.25,1.75)
       (-0.25,0.25) rectangle (0.25,1.75)
       (0.25,-0.25) rectangle (1.75,1.75)
       ;

      \draw[thick] (-1.75,-1.75) rectangle (1.75,1.75)
                   (-0.25,-1.75) rectangle (0.25,1.75) 
                   (-1.75,-0.25) rectangle (1.75,0.25);
      
      \path (-2,1) node {$Z_1$}             
            (-2,0) node {$S_1$}        
            (-2,-1) node {$Y_1$}  
            (-1,2) node {$Z_2$}
            (0,2) node {$S_2$}
            (1,2) node {$Y_2$}
      ;

      \fill (-1,0) circle (2pt)  (1,0) circle (2pt)  (0,0.75) circle
      (2pt) (0,1.25) circle (2pt);

      \path (0,-2.3) node {(c)};
       \end{scope}

    \end{tikzpicture}
    \caption{Uncrossing minimal separations of order $2$}
    \label{fig:cross}
  \end{figure}

  \begin{claim} 
    Let $(A_1,B_1),(A_2,B_2)\in\CT$ be distinct and minimal in $\CT$. Then
    $V(A_1)\subseteq V(B_2)$ and $V(A_2)\subseteq V(B_1)$.

    \proof Let $S_i:=V(A_i)\cap V(B_i)$ and $Y_i:=V(A_i)\setminus S_i$
    and $Z_i:= V(B_i)\setminus S_i$ (see Figure~\ref{fig:cross}(a)). 
    By the minimality of $(A_i,B_i)$, we have
    $Z_i=V\big(C(\CT,S_i)\big)$ and $S_i=N(Z_i)$. Thus $S_1\neq S_2$
    and $Z_1\neq Z_2$, because the two separations are distinct.

    It follows that $(A_1\cup A_2,B_1\cap B_2)$ is a separation with
    $B_1\cap B_2\subset B_i$, and by the minimality of $(A_i,B_i)$
    this separation is not in $\CT$. By
    Lemma~\ref{lem:tangle-closure}(3), this means that its order is at
    least $3$. Thus 
    \begin{equation}
      \tag{$\star$}
      \label{eq:2}
      |S_1\cap Z_2|+|S_1\cap S_2|+|Z_1\cap S_2|=|V(A_1\cup A_2)\cap V(B_1\cap B_2)|\ge 3.
    \end{equation}
    As $|S_i|\le 2$ and $S_1\neq S_2$, it follows that 
    \[
    |S_1\cap Y_2|+|S_1\cap S_2|+|Y_1\cap S_2|=|V(A_1\cap A_2)\cap
    V(B_1\cup B_2)|\le 1.
    \]
    Hence $(A_1\cap A_2,B_1\cup B_2)$ is a separation of order at most
    $1$. As $G$ is 2-connected, the separation is
    not proper, which means that either $V(A_1\cap A_2)=V(G)$ or
    $V(B_1\cup B_2)=V(G)$.  By Lemma~\ref{lem:tangle-closure}(2), we have $(A_1\cap
    A_2,B_1\cup B_2)\in\CT$ and thus $V(A_1\cap A_2)\neq V(G)$. Thus
    $V(B_1\cup B_2)=V(G)$, and this implies $Y_1\cap Y_2=\emptyset$.

    To prove that
    $V(A_i)=S_i\cup Y_i\subseteq V(B_{3-i})=S_{3-i}\cup Z_{3-i}$, we
    still need to prove that $S_i\cap Y_{3-i}=\emptyset$. Suppose for
    contradiction that $S_1\cap Y_2\neq\emptyset$. Then \eqref{eq:2}
    implies $|S_1\cap Y_2|=1$ and $|S_1\cap Z_2|=1$ and $|S_2\cap
    Z_1|=2$ and $S_1\cap S_2=Y_1\cap S_2=\emptyset$ (see
    Figure~\ref{fig:cross}(b)). Note that $(Y_1\cup S_1)\cap Z_2=V(A_1)\setminus
    V(A_2)$. It follows that $(A_1\setminus
    V(A_2), B_1)$ is a separation of $G$ of order $1$, and we have  $(A_1\setminus
    V(A_2), B_1)\in\CT$. Thus $Y_1\cap Z_2=\emptyset$, which implies
    $V(B_2)=Z_2\cup S_2\subset Z_1\cup S_1=V(B_1)$ (see
    Figure~\ref{fig:cross}(c)). This contradicts
    the minimality of $(A_1,B_1)$. Hence  $S_1\cap Y_2=\emptyset$, and
    similarly  $Y_1\cap S_2=\emptyset$.
    \uend
  \end{claim}

  Let $(A_1,B_1),\ldots,(A_m,B_m)$ be an enumeration of all
    minimal separations in $\CT$ of order $2$. Note that there is at
    least one minimal separation of order $2$ even if $G$ has no
    proper separations of order $2$. Thus $m\ge 1$.

    Let $S_i:=V(A_i)\cap V(B_i)$. Then the sets $S_i$ are all distinct,
    because two minimal separations in $\CT$ with the same separators
    are equal. It follows from Claim~1 that $S_i\subseteq V(B_j)$ for
    all $j\in[m]$ and thus
    \[
    S_1\cup\ldots\cup S_m\subseteq X_{\CT}.
    \]
    If $m\ge 2$ this implies $|X_{\CT}|\ge 3$. If $m=1$, then
    $X_{\CT}=V(B_1)$ and thus $|X_{\CT}|\ge 3$ by
    Lemma~\ref{lem:tangle-closure}. 
    \qed
\end{proof}

The results of this section clearly do not extend beyond tangles of
order $3$. For example, the hexagonal grid $H$ in Figure~\ref{fig:hex} has
a (unique) tangle $\CT$ of order $4$. But the set $X_{\CT}$ is empty,
and the graph $H$ has no 3-inseparable set of
cardinality greater than $1$.

Nevertheless, it is
shown in \cite{gro16+} that  there is an extension of the theorem
to tangles of order $4$ if we replace 4-connectivity by the slightly weaker
``quasi-4-connectivity'': a graph $G$ is \emph{quasi-4-connected}
if it is 3-connected and for all separations $(A,B)$ of order $3$,
either $|V(A)\setminus V(B)|\le 1$ or $|V(B)\setminus V(A)|\le 1$. For
example, the hexagonal grid $H$ in Figure~\ref{fig:hex} is
quasi-4-connected. It turns out that there is a one-to-one
correspondence between the tangles of order $4$ and (suitably defined)
quasi-4-connected components of a graph.

\section{A Broader Perspective: Tangles and Connectivity Systems}
\label{sec:tangles}
Many aspects of ``connectivity'' are not specific to
connectivity in graphs, but can be seen in an abstract and much more
general context. We describe ``connectivity'' on some structure as a
function that assigns an ``order'' (a nonnegative integer) to
every ``separation'' of the structure. We study symmetric connectivity
functions, where the separations $(A,B)$ and $(B,A)$ have the same
order. The key
property such connectivity functions need to satisfy is submodularity.

Separations can usually be described as partitions of a suitable set,
the ``universe''. For example, the separations of graphs we considered
in the previous sections are essentially partitions of the edge
set. Technically, it will be convenient to identify a
partition $(\bar X,X)$ with the set $X$, implicitly assuming that
$\bar X$ is
the complement of $X$. This leads to the following definition.

A \emph{connectivity function} on a finite set $U$ is a symmetric and
submodular function
$\kappa\colon 2^U\to\NN$ with $\kappa(\emptyset)=0$. \emph{Symmetric} means that $\kappa(X)=\kappa(\bar
X)$ for all $X\subseteq U$; here and whenever the ground set $U$ is clear from the context we
write $\bar X$ to denote $U\setminus
X$. \emph{Submodular} means that $\kappa(X)+\kappa(Y)\ge\kappa(X\cap
Y)+\kappa(X\cup Y)$ for all $X,Y\subseteq U$. 
The pair $(U,\kappa)$ is sometimes called a \emph{connectivity system}.

The following two examples capture what is known as \emph{edge connectivity}
and \emph{vertex connectivity} in a graph.

\begin{example}[Edge connectivity]
  Let $G$ be a graph. We define the function
  $\nu_G:2^{V(G)}\to\NN$ by letting $\nu_G(X)$ be the number of edges
  between $X$ and $\bar X$. Then $\nu_G$ is a connectivity function on
  $V(G)$.
  \uend
\end{example}

\begin{example}[Vertex connectivity]\label{exa:kappaG}
  Let $G$ be a graph. We define the function $\kappa_G:2^{E(G)}\to\NN$
  by letting $\kappa_G(X)$ be the number of vertices that are incident
  with an edge in $X$ and an edge in $\bar X$. Then $\kappa_G$ is a
  connectivity function on $E(G)$.

  Note that for all separations $(A,B)$ of $G$ we have
  $\kappa_G(E(A))=\kappa_G(E(B))\le\ord(A,B)$, with equality if
  $V(A)\cap V(B)$ contains no isolated vertices of $A$ or $B$.  For
  $X\subseteq E(G)$, let us denote the set of endvertices of the edges
  in $X$ by $V(X)$. Then for all $X\subseteq E(G)$ we have
  $\kappa_G(X)=\ord(A_X,B_X)$, where $B_X=(V(X),X)$ and
  $A_X=(V(\bar X),\bar X)$. The theory of tangles and decompositions
  of the connectivity function of
  $\kappa_G$ is essentially the same as the theory of tangles and
  decompositions of $G$ (partially developed in the previous sections).
  \uend
\end{example}

\begin{example}
    Let $G$ be a graph. For all subsets $X,Y\subseteq V(G)$, we let
  $M=M_G(X,Y)$ be the $X\times Y$-matrix over the 2-element field
  $\mathbb F_2$ with entries $M_{xy}=1\iff xy\in E(G)$.  Now we
  define a connectivity function $\rho_G$ on $V(G)$ by letting
  $
  \rho_G(X)
  $, known as the \emph{cut rank} of $X$, 
   be the row rank of the matrix $M_G(X,\bar X)$. This connectivity
   function was introduced by Oum and Seymour~\cite{oumsey06} to define the \emph{rank
     width} of graphs, which approximates the \emph{clique width}, but
   has better algorithmic properties.
   \uend
\end{example}

Let us also give an example of a connectivity function not related to
graphs.

\begin{example}\label{exa:3}
  Let $M$ be a matroid with ground set $E$ and rank function $r$. (The rank of a set
  $X\subseteq E$ is defined to be the maximum size of an independent set
  contained in $X$.) The connectivity function of $M$ is the set
  function $\kappa_M:E\to\NN$ defined by $\kappa_M(X)=r(X)+r(\bar
  X)-r(E)$ (see, for example, \cite{oxl11}).
  \uend
\end{example}


\subsection{Tangles}

Let $\kappa$ be a connectivity function on a set $U$.
A \emph{$\kappa$-tangle} of order $k\ge0$ is a set $\CT\subseteq 2^U$
satisfying the following conditions.
  \begin{nlist}{T}
  \setcounter{nlistcounter}{-1}
  \item\label{li:t0}
    $\kappa(X)<k$ for all $X\in\CT$, 
  \item\label{li:t1}
    For all $X\subseteq U$ with $\kappa(X)<k$, either $X\in\CT$ or
    $\bar X\in\CT$.
  \item\label{li:t2}
    $X_1\cap X_2\cap X_3\neq\emptyset$ for all $X_1,X_2,X_3\in\CT$.
  \item\label{li:t3}
    $\CT$ does not contain any singletons, that is, $\{a\}\not\in\CT$ for all $a\in U$.
\end{nlist}
We denote the order of a $\kappa$-tangle $\CT$ by $\ord(\CT)$.

We mentioned in Example~\ref{exa:kappaG} that the theory of
$\kappa_G$-tangles is essentially the same as the theory of tangles in
a graph. Indeed, $\kappa_G$-tangles and $G$-tangles are ``almost'' the
same. The following proposition makes this precise.

We call an edge of a graph \emph{isolated} if both of its endvertices
have degree $1$. We call an edge \emph{pendant} if it
is not isolated and has one endvertex of degree $1$. 

\begin{proposition}\label{prop:GvsK}
  Let $G$ be a graph and $k\ge 0$.
  \begin{enumerate}
  \item If $\CT$ is a $\kappa_G$-tangle of order $k$, then 
    \[
    \CS:=\big\{(A,B)\bigmid(A,B)\text{ separation of $G$ of order $<k$
      with
    }E(B)\in\CT\big\}
    \]
    is a $G$-tangle of order $k$.
  \item If $\CS$ is a $G$-tangle of order $k$, then
    \[
    \CT:=\big\{ E(B)\bigmid (A,B)\in\CS\big\}
    \]
    is a $\kappa_G$-tangle of order $k$, unless 
    \begin{eroman}
    \item either $k=1$ and there is an isolated vertex $v\in V(G)$ 
      such that $\CS$ is the set of all separations $(A,B)$ of order
      $0$ with with $v\in V(B)\setminus V(A)$,
    \item or $k=1$ and there is an isolated edge $e\in E(G)$ such that
      $\CS$ is the set of all separations $(A,B)$ of order $0$
      with $e\in E(B)$,
    \item or $k=2$ and there is an isolated or pendant edge
      $e=vw\in E(G)$ and $\CS$ is the set of all separations $(A,B)$
      of order at most $1$ with $e\in E(B)$.
\end{eroman}
  \end{enumerate}
\end{proposition}

We omit the straightforward (albeit tedious) proof.

\begin{figure}
  \centering
  \begin{tikzpicture}[
    vertex/.style={draw,circle,fill=black,inner sep=0mm,minimum
    size=2mm},
  every edge/.style={draw,thick}
  ]
  \path (-2,0) node[vertex] (v1) {}
  (0,-0.75) node[vertex] (v2) {}
  (0,0.75) node[vertex] (v3) {}
  (1.5,-0.75) node[vertex] (v4) {}
  (1.5,0.75) node[vertex] (v5) {}
  (-2,1.5)  node[vertex] (v6) {}
  (-0.7,-2.25) node[vertex] (v7) {}
  (0.7,-2.25) node[vertex] (v8) {};
  
  \draw (v1) edge node[below] {$e_3$} (v2) 
  (v1) edge node[above] {$e_2$} (v3) 
  (v6) edge node[above] {$e_1$} (v3) 
  (v3) edge node[left] {$e_4$} (v2) 
  (v2) edge node[left] {$e_5$} (v7) 
  (v2) edge node[right] {$e_7$} (v8) 
  (v7) edge node[above] {$e_6$} (v8) 
  (v3) edge node[above] {$e_8$} (v5) 
  (v2) edge node[above] {$e_9$} (v4) 
  (v5) edge node[right] {$e_{10}$} (v4) 
;

  \end{tikzpicture}
  \caption{A graph $G$  with three $G$ tangles of order $2$ and two
    $\kappa_G$-tangles of order $2$}
  \label{fig:GvsK}
\end{figure}
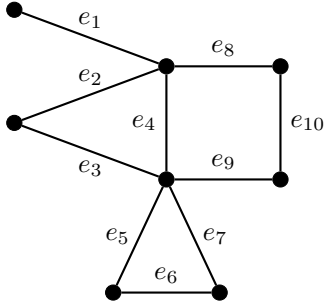

\begin{example}\label{exa:GvsK}
  Let $G$ be the graph shown in
  Figure~\ref{fig:GvsK}. $G$ has one
  tangle of order $1$ (since it is connected) and three tangles of
  order $2$ corresponding to the three biconnected components. The
  $G$-tangle corresponding to the ``improper'' biconnected component
  consisting of the edge $e_1$ and its endvertices does not correspond to
  a $\kappa_G$-tangle (by Proposition~\ref{prop:GvsK}(2-iii)).
  \uend
\end{example}

A \emph{star} is a connected graph in which at most $1$ vertex has
degree greater than $1$. Note that we admit degenerate stars
consisting of a single vertex or a single edge.

\begin{corollary}\label{cor:graph-tangle}
  Let $G$ be a graph that has a $G$-tangle of order $k$. Then $G$ has
 a $\kappa_G$-tangle of order $k$, unless $k=1$ and $G$ only has isolated edges
  or $k=2$ and all connected components of $G$ are stars.
\end{corollary}

\section{Decompositions and Duality}

A \emph{cubic tree} is a
tree where every node that is not a leaf has degree~$3$. An
\emph{oriented edge} of a tree $T$ is a pair $(s,t)$, where $st\in
E(T)$. We denote the set of all oriented edges of $T$ by $\vec E(T)$
and the set of leaves of $T$ by $L(T)$.  A \emph{branch
  decomposition} of a connectivity function $\kappa$ over $U$ is a pair $(T,\xi)$, where $T$ is a cubic
tree and $\xi$ a bijective mapping from $L(T)$ to $U$. For every
oriented edge $(s,t)\in\vec E(T)$ we define $\tilde\xi(s,t)$ to be the
set of all $\xi(u)$ for leaves $u\in L(T)$ contained in  the same
connected component of $T-\{st\}$ as $t$. Note that
$\tilde\xi(s,t)=\bar{\tilde\xi(t,s)}$. We define the \emph{width} of
the decomposition $(T,\xi)$ be the maximum of the values
$\kappa(\tilde\xi(t,u))$ for $(t,u)\in\vec E(T)$. The \emph{branch
  width} of $\kappa$, denoted by $\bw(\kappa)$, is the minimum of the widths of all its branch
decompositions.

The following fundamental result relates tangles and branch
decompositions; it is one of the reasons why tangles are such
interesting objects.

\begin{theorem}[Duality Theorem; Robertson and Seymour~\cite{gm10}]\label{theo:duality}
  The branch width of a connectivity function $\kappa$ equals the
  maximum order of a $\kappa$-tangle.
\end{theorem}

We omit the proof.

Let $G$ be a graph. A \emph{branch decomposition} of $G$ is defined to
be a branch decomposition of $\kappa_G$, and the \emph{branch width} of
$G$, denoted by $\bw(G)$, is the branch width of $\kappa_G$. 

\begin{figure}
  \centering
  \begin{tikzpicture}[
    vertex/.style={draw,circle,fill=black,inner sep=0mm,minimum
    size=1.5mm},
  every edge/.style={draw,thick},
  scale=0.8
  ]
 
    \node[vertex] {}
       child { node[vertex] {}
         child { node {$e_1$} }
         child { node[vertex] {}
           child { node[vertex] {}
             child { node {$e_2$} }
             child { node {$e_3$} }
           }
          child {node {$e_4$} }
         }
       }
       child { node {$e_8$}}
       child { node[vertex] {}
         child { node {$e_{10}$} }
         child { node[vertex] {}
           child { node {$e_9$ } }
           child { node[vertex] {}
             child { node[vertex] {}
               child { node {$e_5$ } }
               child { node {$e_6$} }
             }
             child { node {$e_7$} }
           }
         }
       }
       ;

  \end{tikzpicture}
  \caption{A branch decomposition of width $2$ of the graph shown in Figure~\ref{fig:GvsK}}
  \label{fig:bdec}
\end{figure}
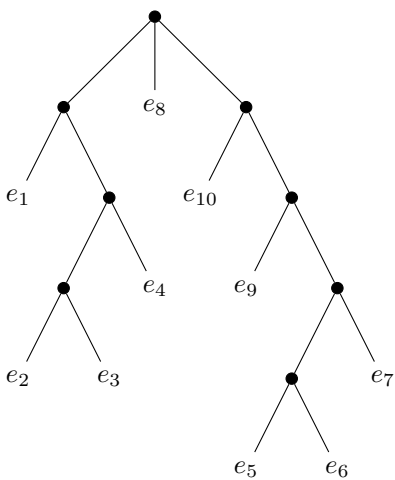

\begin{example}
  Let $G$ be the graph shown in
  Figure~\ref{fig:GvsK}. Figure~\ref{fig:bdec} shows a branch
  decomposition of $G$ of width $2$. Thus $\bw(G)\le 2$. As $G$ has a
  tangle of order $2$ (see Example~\ref{exa:GvsK}), by the Duality
  Theorem we have $\bw(G)=2$.
  \uend
\end{example}

The
branch width of a graph is closely related to the better-known \emph{tree
width} $\tw(G)$: it is not difficult to prove that
\[
\bw(G)\le\tw(G)+1\le\max\left\{\frac{3}{2}\bw(\kappa_G),\;2\right\}
\]
(Robertson and Seymour~\cite{gm10}). Both inequalities are tight. For
example, a complete graph $K_{3n}$ has branch width $2n$ and tree
width $3n-1$, and a path of length $3$ has branch width $2$ and tree
width $1$. There is also a related duality theorem for tree width, due
to Seymour and Thomas~\cite{seytho93}:
$\tw(G)+1$ equals the maximum order of bramble of $G$. (Recall the
characterisation of tangles that we gave in Theorem~\ref{theo:reed}
and the definition of brambles right after the theorem.)

\subsection*{Acknowledgements}
I thank Pascal Schweitzer and Konstantinos Stavropoulos for helpful
comments on a earlier version of the paper.


\end{document}